\documentclass[12pt,a4paper]{article}
\usepackage[utf8]{inputenc}
\usepackage{amsmath}
\usepackage{amsfonts}
\usepackage{amssymb}
\usepackage{mathrsfs}  
\usepackage{amsthm}
\newtheorem{theorem}{Theorem}
\newtheorem{proposition}{Theorem}

\newtheorem{definition}{Definition}

\usepackage{enumitem}
\usepackage{setspace}
\usepackage[backend=bibtex, style=authoryear, useprefix, doi=true, maxbibnames=10]{biblatex}
\bibliography{bib.bib} 
\usepackage{hyperref}

 \hypersetup{
 	pdfauthor={Sven Neth},
	pdftitle={Better Foundations for Subjective Probability},
	pdfkeywords={Representation Theorem, Decision Theory, Bayesian Epistemology, Probability},
	pdflang={en},
	hidelinks
}

\author{Sven Neth \\ Forthcoming in \emph{Australasian Journal of Philosophy}\thanks{This is a preprint of a paper whose final and definite form will be published in the \emph{Australasian Journal of Philosophy}, which is available online at \url{https://www.tandfonline.com/journals/rajp20}. Please cite published version.}}
\date{}
\title{Better Foundations for Subjective Probability}

\begin{document}

\maketitle

\begin{abstract}
How do we ascribe subjective probability? In decision theory, this question is often addressed by representation theorems, going back to \textcite{Ramsey1926}, which tell us how to define or measure subjective probability by observable preferences. However, standard representation theorems make strong rationality assumptions, in particular expected utility maximization. How do we ascribe subjective probability to agents which do not satisfy these strong rationality assumptions? I present a representation theorem with weak rationality assumptions which can be used to define or measure subjective probability for partly irrational agents.
\end{abstract}

\section{Introduction}\label{sec:intro}

In philosophy, psychology and economics, we often ascribe subjective probability or credence to people. For example, we might say that Ann's subjective probability that it will rain tomorrow is .3. What is the basis for such ascriptions of subjective probability? 

If we want to find out Ann's subjective probability that it will rain tomorrow, one natural idea is to look at which bets Ann is willing to accept. If I offer Ann a bet which pays one dollar if it rains tomorrow, how much is this bet worth to her? There is a long tradition in decision theory inspired by this idea, going back to \textcite{Ramsey1926}. In this tradition, we assume that an agent  satisfies certain principles of rationality and then \emph{define} or \emph{measure} subjective probability in terms of preferences.\footnote{\textcite{Buchak2017} calls this `interpretative decision theory': we use decision-theoretic principles to interpret an agent's mental states in a process of `radical interpretation' \parencite{Davidson1973, Lewis1974}.}

But what if Ann does not satisfy the rationality assumptions required by Ramsey? Ramsey assumes that the agent under consideration is an expected utility maximizer and there are good reasons to doubt that real-life agents maximize expected utility. Now perhaps this means that real-life agents do not have any subjective probabilities or that we cannot measure them. However, there is an alternative option: we can provide decision-theoretic foundations for subjective probability with weaker rationality assumptions.

I introduce a representation theorem building on \textcite{Savage1972} and \textcite{Krantz1971} which connects subjective probability to preference with much weaker rationality assumptions than standard representation theorems. In particular, I allow agents to not maximize expected utility, to violate stochastic dominance and to consider most options incomparable. The key idea is to start with comparative probability judgments and to construct a unique probability function which represents these comparative judgments. My representation theorem has important philosophical upshots: it makes sense of how we can ascribe precise subjective probability to partly irrational agents and how decision theory can provide useful advice. 

As suggested above, there are two ways of understanding the project of grounding ascriptions of subjective probability in preferences. First, one might attempt to define subjective probability in terms of preference. On this view, to say that Ann's subjective probability that it will rain tomorrow is .3 just means that Ann is willing to accept certain bets. So ascriptions of subjective probability are a `representational device' to talk about something more fundamental: Ann's preferences. Let us call this view \emph{constructivism}. Second, one might think that while subjective probability is not reducible to preference, we can measure subjective probability by observable preferences. Call this view \emph{realism}. I will mostly remain neutral between constructivism and realism. Like Ramsey's approach, my representation theorem can be interpreted as defining or measuring subjective probability in terms of preference. However, I will later suggest that the theorem naturally fits with an intermediate position between constructivism and realism: comparative probability is psychologically real but numerical probability functions are merely `representational devices' for talking about comparative probability.

Here is the plan. I start by introducing some terminology (\ref{sec:set-up}) and explain Ramsey's method for measuring subjective probability (\ref{sec:ramsey-method}). I discuss some problems for Ramsey's method (\ref{sec:ramsey-problem}). I introduce better foundations for subjective probability (\ref{sec:better-foundations}) and explain how they overcome the problems for Ramsey's method (\ref{sec:problems-solved}). I finish by sketching two ideas suggested by the representation theorem (\ref{sec:interpretation}): the view that comparative probability is more fundamental than numerical probability and a subjectivist version of the classical interpretation of probability.

\section{Set-up}\label{sec:set-up}

We have a set $\Omega$ of \emph{states} which describe the world apart from our agent's preferences and a $\sigma$-algebra $\mathscr{F}$ of subsets of $\Omega$ which are called \emph{events}.\footnote{A \textit{$\sigma$-algebra} on $\Omega$ is a set of subsets of $\Omega$ which contains $\Omega$ and is closed under complementation and countable unions.} For any $X \in \mathscr{F}$, we denote the relative complement of $X$ in $\Omega$ by $X^C$. We have a set $\mathscr{O}$ of \emph{outcomes} which contain everything our agent cares about. 

Following \textcite{Savage1972}, \emph{acts} are functions from states to outcomes. An act is \emph{finite-valued} iff it only takes finitely many different outcomes. Our \emph{act space} $\mathscr{A}$ is the set of all finite-valued acts. This means that we can write each act $f \in \mathscr{A}$ as $\{o_{1},E_{1};...;o_{n}, E_n\}$, where events $E_{1},....,E_{n}$ are pairwise disjoint sets and their union is  $\Omega$ and for each $E_{i}$ with $1\leq i\leq n$, $o_{i}$ is the unique outcome $o_{i}\in \mathscr{O}$ such that $f(\omega)=o_{i}$ for all $\omega\in E_{i}$.

I write `$\succsim$' for our agent's preference ordering over acts, a binary relation on $\mathscr{A}$. The intended interpretation of $f \succsim g$ is that our agent weakly prefers $f$ to $g$. Strict preference $(\succ)$ and indifference $(\sim)$ are defined in the usual way.\footnote{So $f \succ g \iff (f \succsim g) \land \neg (g \succsim f)$ and  $f \sim g \iff (f \succsim g) \land (g \succsim f).$} For each outcome $o \in \mathscr{O}$, the \emph{constant act yielding $o$}, written $\underline{o}$, is the act which assigns $o$ to all $\omega \in \Omega$. I define, for any $o, o' \in \mathscr{O}$, $o  \succsim  o'$ iff $\underline{o} \succsim \underline{o'}$. I use the term `option' to talk about both acts and outcomes.

\section{Ramsey's Method}\label{sec:ramsey-method}

\textcite{Ramsey1926} proposes axioms on preferences which imply that our agent is representable as expected utility maximizer.\footnote{\textcite[ch. 3]{Jeffrey1990} and \textcite{Bradley2004} reconstruct Ramsey's reasoning. Ramsey originally used a different  framework which does not distinguish states and outcomes, while I am reconstructing Ramsey's reasoning in terms of the Savage framework.} This means that there is some probability function and utility function such that our agent always prefers acts with higher expected utility.\footnote{Relative to probability function $p : \mathscr{F} \to [0,1]$ and utility function $u : \mathscr{O} \to \mathbb{R}$, the \emph{expected utility} of act $f = \{o_1,E_1; .... ; o_n,E_n\}$ is $\mathbb{E}_{u,p}(f) = \sum_{i=1}^n p(E_i)u(o_i).$} Ramsey then gives us a way to construct or infer our agent's utility function without already knowing our agent's probability function.

Once we have the utility function, Ramsey pins down the subjective probability of any event $E \in \mathscr{F}$ as follows. First, we find three outcomes $b,m,w \in \mathscr{O}$ (best, medium and worst) such that our agent strictly prefers $b$ over $w$ and is indifferent between getting $m$ for certain and a bet which yields $b$ if $E$ happens and $w$ otherwise:
\begin{equation}\label{r1}
b \succ w,
\end{equation}
\begin{equation}\label{r2}
m \sim \{b, E; w, E^{C}\}.
\end{equation}
Then we use the assumption that our agent maximizes expected utility to infer that $p(E) = \frac{u(m) - u(w)}{u(b) - u(w)}$.\footnote{Proof: since our agent maximizes expected utility, (\ref{r1}) and (\ref{r2}) entail that $u(m) = u(w) + p(E)(u(b) - u(w))$. Therefore, $u(m) - u(w) = p(E)(u(b) - u(w))$, so $p(E) = \frac{u(m) - u(w)}{u(b) - u(w)}$. The utility function is unique up to positive affine transformation so $p(E)$ is unique.} So for Ramsey, subjective probabilities are betting odds. Since event $E$ was arbitrary, we can use Ramsey's method to uniquely pin down the subjective probability of all events. Depending on whether we accept constructivism or realism, we can think of this as a definition of subjective probability in terms of preferences or as a way to measure subjective probability by preferences.

Ramsey had a lasting influence on decision theory.\footnote{\textcite{Fishburn1981} provides a great survey of decision theory after Ramsey. \textcite{Misak2020} places Ramsey's work in its broader intellectual context.} \textcite{Savage1972} also lays down axioms on the preference relation and proves a representation theorem which shows that any agent obeying these axioms can be represented as expected utility maximizer with a unique probability function. Many of the problems for Ramsey's approach I will discuss below generalize to Savage's representation theorem. However, as we will see later, the work of Savage holds key insights for an alternative approach to measure subjective probability.\footnote{\textcite{Jeffrey1990} develops a different framework for decision theory in which states, acts and outcomes are all propositions. However, Jeffrey's axioms do not pin down a unique probability function.}

\section{Problems for Ramsey}\label{sec:ramsey-problem}

I now turn to explain why Ramsey's method is not an adequate foundation for subjective probability.

\subsection{Strong Rationality Assumptions}\label{sec:strong-rationality}

Ramsey assumes that the agent under consideration is an expected utility maximizer. However, there are good reasons to think that real-life agents are not expected utility maximizers. Suppose you only care about money and choose between the following two lotteries:\footnote{A \emph{lottery} is a probability distribution over outcomes and can be realized by multiple acts.}

\begin{enumerate}

\item One million dollars for certain.\label{L1}

\item 89 \% chance of winning one million dollars, 10 \% chance of winning five million dollars, 1 \% chance of winning nothing.\label{L2}

\end{enumerate}

\noindent You also choose between the following two lotteries:

\begin{enumerate}[resume] 

\item 89 \% chance of winning nothing, 11 \% chance of winning one million dollars.\label{L3}

\item 90 \% chance of winning nothing, 10 \% chance of winning five million dollars.\label{L4}

\end{enumerate}

If you strictly prefer (\ref{L1}) over (\ref{L2}) and (\ref{L4}) over (\ref{L3}), your preferences are incompatible with expected utility maximization \parencite{Allais1953}.\footnote{ If you strictly prefer (\ref{L4}) over (\ref{L3}), $.1u(\$\textrm{5 Million}) > .11 u(\$\textrm{1 Million})$. So, adding the same term on both sides, $.1u(\$\textrm{5 Million}) + .89 u(\$\textrm{1 Million}) > .11 u(\$\textrm{1 Million}) + .89 u(\$\textrm{1 Million})$, which means that  $.1u(\$\textrm{5 Million}) + .89 u(\$\textrm{1 Million}) > u(\$\textrm{1 Million})$. But if you strictly prefer (\ref{L1}) over (\ref{L2}), $u(\$\textrm{1 Million}) > .89 u(\$\textrm{1 Million}) + .1 u(\$\textrm{5 Million})$.} However, real-life agents sometimes exhibit this pattern of preferences \parencite{oliver2003}. My point here is not that this pattern of preferences is rationally permissible, as argued by \textcite{Buchak2013}. Rather, my point is that real-life agents apparently have such `Allais-preferences'. Therefore, we cannot use Ramsey's method to define or measure their subjective probabilities. However, it still seems like such agents have subjective probabilities---after all, they are presumably \emph{using} their subjective probabilities to reason that (\ref{L1}) is better than (\ref{L2}) and  (\ref{L4}) is better than (\ref{L3}). So Ramsey's method is not a good foundation for subjective probability. To make this vivid, imagine you find out that Ann has the preferences described above. Should you conclude that Ann cannot have any subjective probabilities or that there is no way for us to find out what these probabilities are? I think not.

People also sometimes choose \emph{stochastically dominated} options. Option $A$ stochastically dominates option $B$ if for every outcome $o \in \mathscr{O}$, the probability that $A$ yields an outcome weakly preferred to $o$ is greater than or equal to the probability that $B$ yields an outcome weakly preferred to $o$. It is generally agreed that you should:

\begin{quote}
\textbf{Respect Stochastic Dominance.} If $f$ stochastically dominates $g$, then $f \succsim g$.
\end{quote}

This principle follows from many normative decision theories, such as expected utility theory, risk-weighted expected utility and others.\footnote{\textcite{Buchak2013} and \textcite{Tarsney2020} defend stochastic dominance. \textcite{Bader2018} points out the wide applicability of stochastic dominance reasoning even if outcomes are incomparable. \textcite{RussellForth} discusses some problems arising in such a setting.} However, empirical studies show that people sometimes violate this principle. Consider the following two lotteries:

\begin{enumerate}
\setcounter{enumi}{4}

\item 5\% chance of \$12, 5\% chance of \$14, 90\% chance of \$96.\label{L5}

\item 10\% chance of \$12, 5\% chance of \$90, 85\% chance of \$96.\label{L6}

\end{enumerate}

It is not hard to see that (\ref{L5}) stochastically dominates (\ref{L6}).\footnote{Both lotteries are sure to pay at least \$12. The probability of winning at least \$14 is 95\% in (\ref{L5}) and 90\% in (\ref{L6}). The probability of winning at least \$90 is 90\% in (\ref{L5}) and (\ref{L6}).} Nonetheless, in a study reported by \textcite{Birnbaum1998}, most subjects chose (\ref{L6}) over (\ref{L5}). This is presumably because they rely on quick but imperfect heuristics in their decision making. Again, my claim is not that these preferences are rational but only that real-life agents have such preferences. Therefore, we cannot use Ramsey's method to define or measure their subjective probabilities. But again, while it might be irrational to have preferences which violate stochastic dominance, such preferences do not seem to preclude agents from having subjective probabilities.

Finally, people sometimes regard options as \emph{incomparable} in value. Consider, for example, the choice between a career as a doctor and a career as a rock star. Both career choices can lead to a fulfilling and valuable life. However, what makes them valuable is radically different. It is difficult to see how one could compare the two options. Someone could reasonably think that one is not better than the other but neither are they exactly equally good \parencite{Chang2002}.

Incomparability arises both at the level of outcomes and acts. It is natural to understand the career choice example in terms of incomparable outcomes. In contrast, a second kind of incomparability might arise because it is too difficult to compare acts even if all of their outcomes are comparable. Suppose, for example, that you like more money rather than less. When faced with a choice between two complicated investment portfolios, you might nonetheless not have any preference between them because it is too difficult for you to reason about the decision problem. Again, my claim is not that incomparability is rational, but merely that real-life agents sometimes have such preferences.\footnote{Many have defended the stronger claim that incomparability can be rational. \textcite[p.~102]{Joyce1999} writes that `a decision maker can be perfectly rational even when her preferences do not satisfy the completeness axiom'. \textcite[p.~446]{Aumann1962} writes that `[o]f all the axioms of utility theory, the completeness axiom is perhaps the most questionable'. Similar points are made by \textcite{Hare2010, Bales2014, Schoenfield2014, Bader2018, Sen2018}.}

Expected utility theory has no room for incomparability. This is because your utility function assigns a real number to each outcome and so renders all outcomes comparable. Each act is ranked by its expected utility so all acts are comparable as well. Since real-life agents sometimes regard both outcomes and acts as incomparable, their preferences cannot be represented as expected utility maximization. However, it is not plausible that incomparability precludes agents from having subjective probabilities.

Here is the upshot. There are good reasons to think that real-life agents are not expected utility maximizers. Therefore, we cannot use Ramsey's method to define or measure their subjective probabilities. Some might take this as reason to embrace a kind of nihilism: such agents do not have subjective probabilities or there is no way to measure what they are. A better response is to provide foundations for subjective probability which apply even to agents which fail to maximize expected utility, do not respect stochastic dominance and consider some options incomparable. One might still think that we should model the beliefs of real-life agents by something other than precise probability functions. However, we can make room for irrational preferences without giving up decision-theoretic foundations for precise subjective probability. 

\subsection{No Useful Advice}

The standard decision-theoretic advice is to maximize expected utility relative to your subjective probability function and utility function. For this advice to be useful, we first need to figure out what your probability function \emph{is}.\footnote{The same point applies to the utility function but I focus on subjective probability.} However, if we define or measure your probability function on the assumption that you maximize expected utility, the advice to maximize expected utility can never be useful. Therefore, decision theory cannot play the role of giving useful advice.\footnote{\textcite[p.~99]{Resnik1987} writes, about representation theorems in decision theory: `the theorem can be applied only to those agents with a sufficiently rich preference structure; and if they have such a structure, they will not need utility theory---because they will already prefer what it would advise them to prefer'. \textcite{Meacham2011} and \textcite{Easwaran2014} deploy similar arguments. \textcite{Beck2021} also discuss the question of how decision-theoretic models can provide useful advice.}

This puzzle arises on both constructivism and realism. For constructivists, your probability function is defined in terms of preferences which satisfy certain axioms. If you violate the axioms, you simply do not \emph{have} a probability function and the advice to maximize expected utility is meaningless. For realists, you might still have a probability function if you violate the axioms. However, standard representation theorems give us no way to infer what this probability function is, so we cannot use decision theory to give useful advice.\footnote{This puzzle also arises for non-standard decision theories such as risk-weighted expected utility theory \parencite{Buchak2013} and weighted-linear utility theory \parencite{BottomleyForth}. On these theories, we also need an independent grasp on your subjective probability function, your utility function and possibly other functions like your risk function in order for the theory to provide useful advice.}

One reaction to this problem is to say that the only advice decision theory provides is: `obey the axioms!'. On this view, decision theory is merely a theory of consistency \parencite{Dreier1996, Okasha2016}. While I have no knock-down objection to this position, it is unattractive because it makes decision theory largely irrelevant to non-ideal agents like us who are pretty much guaranteed to violate some normative principle of decision making. It would be better if we could make sense of how decision theory can provide useful advice to partly irrational agents. As I will show, we can indeed make sense of this, which considerably weakens the plausibility of this response.

\subsection{Dependence on Utility}

Ramsey defines subjective probability as ratio of utilities. This requires a very rich space of outcomes. For Ramsey, outcomes must allow for continuous gradations of value.\footnote{For example, \textcite[p.~152]{Fishburn1981} writes that in Ramsey's approach, the set of outcomes `must be infinite and give arbitrarily fine gradations in utility'.} However, it seems like agents can have subjective probabilities while not making such fine-grained distinctions of value. We could even imagine agents who do not have a utility function at all but merely an ordinal ranking of outcomes. For example, we can imagine an agent which only distinguishes between two outcomes, GOOD and BAD. Ramsey must deny that such an agent could have subjective probabilities or that we can find out what they are. This seems implausible.

More broadly, Ramsey's approach gives utility a certain kind of priority over subjective probability. But you might think that subjective probability is conceptually independent of utility. It would be great to disentangle the assumptions needed to measure subjective probability from strong assumptions about the structure of value, for example that the value of all outcomes is comparable and that value can be measured by a real-valued utility function. Such assumptions about value have seemed implausible to many philosophers and it would be great to have foundations for subjective probability which do not rely on them.

\section{Better Foundations}\label{sec:better-foundations}

We can provide better foundations for subjective probability. I introduce and explain a representation theorem building on \textcite{Savage1972} and \textcite{Krantz1971} which yields a unique probability function representing our agent's beliefs on weak rationality assumptions. The key idea is to start with \emph{comparative probability judgments} and to construct a unique probability function which represents these comparative judgments.

\subsection{Comparative Probability}

What does it mean to think that one event is more probable than another? \textcite{Savage1972} proposes to define comparative probability judgments in terms of certain kinds of preferences. Suppose our agent strictly prefers outcome $b$ over outcome $w$. Now the intuition is that our agent \emph{prefers the better prize on the more probable event}. So if our agent prefers the act $\{b, X; w, X^{C}\}$ over $\{b, Y; w, Y^{C}\}$, this means that our agent believes that event $X$ is at least as likely as event $Y$, written $X \succcurlyeq Y$. So we can use acts of the form $\{b, X; w, X^{C}\}$, where $b \succ w$, to define or infer our agent's comparative probability judgments. Let us call these  \emph{test acts}. 

We define the following relation $\succcurlyeq$ on $\mathscr{F}$: 
\begin{definition}\label{def:comparative} 
$X \succcurlyeq Y$ iff $\{b, X; w, X^{C}\} \succsim \{b, Y; w, Y^{C}\}$ for some $b,w \in \mathscr{O}$ with $b \succ w$.
\end{definition} 
Strict comparative probability ($\succ$) and indifference ($\approx$) are defined in the standard way.\footnote{$X \succ Y \iff (X \succcurlyeq Y) \land \neg (Y \succcurlyeq X)$ and $X \approx Y \iff (X \succcurlyeq Y) \land (Y \succcurlyeq X)$. In a slight abuse of notation, I use the same symbol ($\succ$) for both strict preference and strict comparative probability. This way to link comparative probability to preferences is standard \parencite{Fishburn1986, Icard2016}.}

We can understand Savage's proposal in two ways. For constructivists, comparative probability reduces to preferences. (This is Savage's own view.) For realists, comparative probability does not reduce to preferences, but we can use preferences to measure comparative probability judgments.  The core proposal of this paper is compatible with both ways of understanding comparative probability. However, I think that realism about comparative probability is ultimately more plausible and can tell a better story about some of the axioms below.  I will return to this issue later.

\subsection{Axioms}

The first axiom ensures that the comparative probability ordering does not depend on our particular choice of outcomes:
\begin{quote}
\textbf{Outcome Independence}. For all $X, Y \in \mathscr{F}$, if $\{b, X; w, X^{C}\} \succsim \{b, Y; w, Y^{C}\}$ for some $b,w \in \mathscr{O}$ such that $b \succ w$, then $\{b, X; w, X^{C}\} \succsim \{b, Y; w, Y^{C}\}$ for all $b,w \in \mathscr{O}$ such that $b \succ w$.
\end{quote}
You would violate this axiom if you prefer to bet one dollar on event $X$ rather than event $Y$ but you also prefer to bet two dollars on $Y$ rather than $X$. In this case, we cannot elicit stable comparative probability judgments from your preferences.

The next axiom demands that our agent is not indifferent among all outcomes:
\begin{quote}
\textbf{Non-Degeneracy}. There are outcomes $b,w \in \mathscr{O}$ with $b \succ w$.
\end{quote}
What is the status of this axiom? Does the existence of subjective probability really require that you are not indifferent between all outcomes? \textcite{Eriksson2007} point out that we can imagine a Zen monk who is indifferent between all outcomes but nonetheless has subjective probabilities. Thus, there are problems with \textbf{Non-Degeneracy} understood along constructivist lines. However, if we are realists, we can accept \textbf{Non-Degeneracy} as a condition under which we can measure comparative probability. The Zen monk might have subjective probabilities, but if they are really indifferent among everything, there is simply no way for us to find out what these subjective probabilities are. So we can think of this axiom as a \emph{structure axiom} which ensures that preferences are rich enough to measure subjective probability.\footnote{\textcite[p.~82]{Joyce1999} distinguishes structure axioms and rationality axioms.}

Here is the third axiom:
\begin{quote}
\textbf{Restricted Ordering}. The relation $\succsim$ restricted to test acts with the same outcomes is complete and transitive. This means for any $b,w \in \mathscr{O}$ with  $b \succ w$, for all $X,Y \in \mathscr{F}$ we have either  $\{b, X; w, X^{C}\} \succsim \{b, Y; w, Y^{C}\}$ or $\{b, Y; w, Y^{C}\} \succsim \{b, X; w, X^{C}\}$.\footnote{Two test acts with different outcomes needn't be comparable if the outcomes themselves are incomparable. Thanks to an anonymous referee for bringing this issue to my attention.} And if $\{b, X; w, X^{C}\} \succsim \{b, Y; w, Y^{C}\}$ and $\{b, Y; w, Y^{C}\} \succsim \{b, Z; w, Z^{C}\}$, then $\{b, X; w, X^{C}\} \succsim \{b, Z; w, Z^{C}\}$.
\end{quote}
I do not constrain the preference relation in general to be complete and transitive, which leaves room for incomparability.

Why accept \textbf{Restricted Ordering}? Given our definition of comparative probability, \textbf{Restricted Ordering} requires that the comparative probability judgments of our agent are complete and transitive. There are reasons to be skeptical of both.\footnote{\textcite[p.~339]{Fishburn1986} discusses examples in which comparative probability judgments violate completeness and transitivity. \textcite{Ding2021} study logics for comparative probability without completeness.} For proponents of imprecise credences, rejecting completeness is particularly natural. Perhaps you have some opinion about how likely it is that there is life on mars and some opinion about how likely it is to rain tomorrow but no opinion about which is more likely. This seems particularly plausible if we consider agents which are not perfectly rational.

In response, remember that I want to explain how it is possible to ascribe \emph{precise} credences to agents with some irrational preferences. For this reason, I will not consider agents whose comparative probability judgments fail to be complete and transitive. Such agents fall outside of the scope of my project.

The next two axioms are where the main action is. Let us begin with:  
\begin{quote}
    \textbf{Certain Prize}. For any $b,w \in \mathscr{O}$, if $b \succ w$, then for any $X \in \mathscr{F}$, $\underline{b} \succsim \{b, X; w, X^C \}$ and $\{b, X ; w, X^C \} \succsim \underline{w}$.
\end{quote}
This principle states a plausible minimal rationality condition. It says that if you strictly prefer $b$ to $w$, then you must weakly prefer getting $b$ for certain to an act which yields $b$ if $X$ happens and $w$ otherwise. Further, you must weakly prefer this act to getting $w$ for certain. 

While \textbf{Certain Prize} is quite weak, it is possible to imagine agents which violate this axiom. For example, agents might prefer a risky option over a sure thing because they enjoy the thrill of gambling.\footnote{Thanks to an anonymous referee for raising this objection.} Relatedly, \textbf{Certain Prize} might be violated by agents who prefer randomization \parencite{Icard2021}. One response to this concern is to make more fine-grained distinctions among outcomes \parencite{Dreier1996}. For example, a prize obtained for sure would be a different outcome from the same prize obtained by a risky gamble. However, this move threatens to trivialize decision-theoretic norms. So it is best to concede that while the axiom is weaker than rationality axioms in standard representation theorems, it still makes substantive demands which some agents might violate.

Do agents which violate \textbf{Certain Prize} not have subjective probabilities? This is not very plausible. After all, it is precisely their subjective probabilities which lead them to prefer the risky option. It is more plausible to think that if agents love the thrill of gambling, it might be difficult to determine their subjective probabilities from their preferences. As I show below, \textbf{Certain Prize} is a necessary condition for the agent's comparative probability judgments to be representable by a probability function, so measuring the subjective probability of agents which violate \textbf{Certain Prize} would require a fundamentally different approach to measuring subjective probability.

Here is another key axiom:
\begin{quote}
    \textbf{Alternative Prize}. For any $X,Y,Z \in \mathscr{F}$ and $b, w \in \mathscr{O}$, if $b \succ w$ and $Z$ is such that $X \cap Z = Y \cap Z = \varnothing$, then $\{b, X; w, X^{C}\} \succ \{b, Y; w, Y^{C}\}$ iff $\{b, X \cup Z; w, (X \cup Z)^{C}\} \succ \{b, Y \cup Z; w, (Y \cup Z)^{C}\}$.
\end{quote}

\textbf{Alternative Prize} says the following. Suppose you strictly prefer $b$ over $w$ and you prefer $\{b, X; w, X^{C}\}$ over $\{b, Y; w, Y^{C}\}$. Now we modify both acts as follows: You also get $b$ if some event $Z$ disjoint from both $X$ and $Y$ happens. Now you should prefer $\{b, X \cup Z; w, (X \cup Z)^{C}\}$ to $\{b, Y \cup Z; w, (Y \cup Z)^{C}\}$. This reasoning also works backwards. \textbf{Alternative Prize} has a clear interpretation in terms of probability. If you prefer $\{b, X; w, X^{C}\}$ to $\{b, Y; w, Y^{C}\}$, you think that $X$ is at least as likely as $Y$. Therefore, $X \cup Z$ must be at least as likely as $Y \cup Z$ given that $Z$ is disjoint from both $X$ and $Y$. So you should prefer $\{b, X \cup Z; w, (X \cup Z)^{C}\}$ to $\{b, Y \cup Z; w, (Y \cup Z)^{C}\}$ since you want the better prize on the more probable event.

You might violate \textbf{Alternative Prize} if you have credences which are not additive and represented by an alternative formalism like Dempster-Shafer functions or ranking theory.\footnote{\textcite{Ellsberg1961} gives an example of preferences which violate \textbf{Alternative Prize}.} But my goal is to provide foundations for ascribing subjective \emph{probability} to partly irrational agents. So agents modeled by such formalisms fall outside the scope of my project.\footnote{\textcite[Ch.~14.3]{Titelbaum2022} gives a brief introduction to Dempster-Shafer functions and \textcite{Spohn2012} discusses ranking theory. As both authors note, it is  unclear how these alternatives to probability interact with decision making, which is a  reason to set them aside for our purposes. Thanks to an anonymous referee for the suggestion to consider these frameworks.} It would be desirable to have more general foundations for measuring belief which apply to agents with non-probabilistic credences, but I will not consider such agents here.

My axioms on the preference relation are necessary and sufficient for the comparative probability ordering to be a \emph{qualitative probability} \parencite{DeFinetti1931}:

\begin{definition}
A binary relation $\succcurlyeq$ on $\mathscr{F}$ is a \emph{qualitative probability} iff for all $X, Y, Z \in \mathscr{F}$:
\begin{enumerate}
    \item $\succcurlyeq$ is complete and transitive (Ordering),
    \item $\Omega \succcurlyeq X \succcurlyeq \varnothing$ (Boundedness),
    \item $\Omega \succ \varnothing$ (Non-Triviality),
    \item if $X \cap Z = Y \cap Z = \varnothing$, then $X \succ Y \iff X \cup Z \succ Y \cup Z$ (Qualitative Additivity).
\end{enumerate}
\end{definition}

\begin{theorem}\label{thm:qualitative-probability}
The preference relation $\succsim$ satisfies \textbf{Outcome Independence}, \textbf{Non-Degeneracy}, \textbf{Restricted Ordering}, \textbf{Certain Prize} and  \textbf{Alternative Prize} if and only if the comparative probability ordering $\succcurlyeq$ is a qualitative probability.
\end{theorem}

A proof is provided in the appendix. I follow Savage's definition of comparative probability. Savage also assumes \textbf{Outcome Independence} and \textbf{Non-Degeneracy}. The key difference is that Savage uses much stronger axioms to derive the result that the comparative probability ordering is a qualitative probability. Instead of \textbf{Restricted Ordering}, Savage assumes that preferences are complete and transitive, which rules out incomparable options. This strong assumption is unnecessary to establish that the comparative probability ordering is complete and transitive. It suffices to assume that a small fragment of the preference relation is complete and transitive.

Further, Savage appeals to the `Sure-Thing Principle' in order to establish that the comparative probability ordering satisfies Boundedness, Non-Triviality and Qualitative Additivity. The Sure-Thing-Principle is a strong axiom which rules out the Allais-preferences discussed earlier and plays a crucial role in establishing the existence of an expected utility representation. The key observation is that we can replace the Sure-Thing-Principle by the much weaker rationality axioms \textbf{Certain Prize} and \textbf{Alternative Prize} and still show that the comparative probability ordering is a qualitative probability.\footnote{\textcite{Machina1992} also weaken Savage's axiom to give a `more robust definition of subjective probability'. However, their axioms are stronger than the ones given here, as they entail that preferences always respect stochastic dominance---a property they refer to as `probabilistic sophistication'---and that preferences are complete. My representation theorem shows how to define subjective probability \textit{without} probabilistic sophistication (and without completeness). \textcite{Elliott2017} provides a representation theorem for `frequently irrational' agents and uses a restricted class of two-outcome acts to construct a unique credence and utility function. However, this credence function is not necessarily a probability function, so this approach does not provide foundations for ascribing subjective \emph{probability} to partly irrational agents.}

\textcite[pp.~208-11]{Krantz1971} prove a similar result.\footnote{\textbf{Outcome Independence} is equivalent to the first axiom by \textcite{Krantz1971}, \textbf{Certain Prize} is equivalent to their second axiom and \textbf{Alternative Prize} is equivalent to their third axiom. They also mention \textbf{Non-Degeneracy}.} But instead of \textbf{Restricted Ordering}, they assume that preferences are complete and transitive, which rules out incomparable options. Furthermore, I have shown that my axioms are not only sufficient but necessary for the comparative probability ordering to be a qualitative probability. So my result is a strengthening of \textcite{Krantz1971}, maximally paring down the axioms on the preference relation required to show that the comparative probability ordering is a qualitative probability.

One could also axiomatize comparative probability directly and argue that the qualitative probability axioms are reasonable constraints on belief without trying to justify them by more fundamental axioms about preferences \parencite[p.~91]{Joyce1999}. However, my project is to show how we can infer credences from preferences without already assuming that we have access to comparative probability judgments. Therefore, I start with axioms on the preference relation.

So far, we have seen how preference reveals qualitative probability. How do we get from qualitative probability to quantitative probability? Probability function $p$ \textit{represents} the qualitative probability $\succcurlyeq$ if for all $X, Y \in \mathscr{F}$,
\begin{equation*}
p(X) \geq p(Y) \iff X  \succcurlyeq Y.
\end{equation*}
The axioms introduced so far are necessary but not sufficient for the existence of a probability function representing our qualitative probability \parencite{Kraft1959}. To get around this problem, I add an axiom which ensures that the space of events is sufficiently rich to pin down a (unique) probability function. Here is Savage's proposal:
\begin{quote}
\textbf{Event Richness}. For all $X,Z \in \mathscr{F}$ and outcomes $b,w \in \mathscr{O}$ with $b \succ w$, if $\{b, X; w, X^C \} \succ \{b,Z; w,Z^C\}$, there is a finite partition $\mathcal{Y} = \{Y_1, ... , Y_n\}$ of $\Omega$ such that for all $Y_i \in \mathcal{Y}$, $\{b, X; w, X^C \} \succ \{b, (Z \cup Y_i);w,(Z \cup Y_i)^C\}.$
\end{quote}
This axiom says that we can cut up events very finely. If you strictly prefer the good prize on $X$ rather than $Z$, there is a finite partition of our state space such that you still prefer the good prize on $X$ rather than $Z$ or one of the elements of our partition. It is instructive to state \textbf{Event Richness} in terms of comparative probability. In these terms, it says that if $X \succ Z$, then there exists a finite partition $\mathcal{Y} = \{Y_1, ... , Y_n\}$ of $\Omega$ such that for all $Y_i \in \mathcal{Y}$, $X \succ Z \cup Y_i$.

Why accept \textbf{Event Richness}? \textcite[p.~38]{Savage1972} gives the following argument. Suppose you judge $X$ to be more probable than $Z$. Savage points out that we could plausibly choose a coin and throw it sufficiently often such that you would still judge $X$ to be more probable than $Z$ or any particular sequence of heads and tails. As Savage notes, this doesn't require that you consider the coin to be fair. The possible outcomes of the coin flip form the required partition.

Let us end by briefly reflecting on the plausibility of \textbf{Event Richness}. Does rationality require that you cut up events very finely? Despite Savage's argument, this does not seem very plausible. Like \textbf{Non-Degeneracy}, we should think about \textbf{Event Richness} not as rationality axiom but rather as structure axiom which ensures that preferences are rich enough to fix subjective probability. This means that realism can tell a more plausible story about this axiom than constructivism. According to the realist story, it is not the existence of subjective probability which requires such a rich event space. Rather, the rich event space is necessary to infer (precise) probability from preference. 

\textbf{Event Richness} implies that $\Omega$ is infinite.\footnote{Proof sketch: Assume $\Omega$ is finite. Then consider the least probable event $X$ such that $X \succ \varnothing$. \textbf{Event Richness} demands that there exists a finite partition $\mathcal{Y}$ of $\Omega$ such that for all $Y_i \in \mathcal{Y}$, $\{b, X; w, X^C \} \succ \{b, Y_i;w, Y_{i}^C\}$, so $X \succ  Y_i \succ \varnothing$, which contradicts our assumption.} This might strike you as problematic because it seems possible to have subjective probabilities with a finite state space. One option is to look for another structure axiom which is compatible with finite state spaces but still allows us to derive a unique probability function. As \textcite{Luce1967} and \textcite{Fishburn1986} point out, there are such axioms, but they are rather complicated and do not have intuitive plausibility of \textbf{Event Richness}. Since we need some structure axiom anyways, it seems best to stick with \textbf{Event Richness} because of its intuitive plausibility. However, finding a good replacement for \textbf{Event Richness} which is compatible with finite state spaces is a way in which the representation theorem could be improved.\footnote{Another option would be axioms ensuring that the comparative ordering can be represented by some probability function which needn't be unique \parencite{Scott1964}. However, this is not compatible with providing decision-theoretic foundations for \emph{precise} subjective probability and so I will set it aside. One could also argue that comparative probability orderings on finite spaces should be \emph{extendable} to orderings on infinite state spaces which satisfy \textbf{Event Richness}.}

\subsection{Representation Theorem}

The axioms allow us to prove:
\begin{theorem}\label{thm:representation}
If the preference relation $\succsim$ satisfies \textbf{Outcome Independence}, \textbf{Non-Degeneracy}, \textbf{Restricted Ordering}, \textbf{Certain Prize}, \textbf{Alternative Prize} and \textbf{Event Richness}, there is a unique finitely additive probability function $p : \mathscr{F} \to [0,1]$ representing the comparative probability ordering $\succcurlyeq$, so for all $X,Y \in \mathscr{F}$,
\begin{equation*}
p(X) \geq p(Y) \iff X \succcurlyeq Y.
\end{equation*}
\end{theorem}

Once we have shown that the comparative probability ordering is a qualitative probability, the rest of the proof is due to Savage. Here is a quick proof sketch inspired by \textcite[pp.~120-125]{Kreps1988}:

\begin{proof} The axioms entail that for any $n \in \mathbb{N}$, there is a partition $\mathcal{Y}$ of $\Omega$ into $n$ equiprobable events: events such that $Y_i \approx Y_j$ for each $Y_i, Y_j \in \mathcal{Y}$.\footnote{\textcite[pp.~195-8]{Fishburn1970} provides a detailed reconstruction of this step of the proof. \textcite{Gaifman2018} discuss how it relies on the assumption that the events form a $\sigma$-algebra.}  We write $C(k,n)$ for a union of $k$ cells of this partition. We define, for any $X \in \mathscr{F}$:
\begin{equation*}
k(X,n) = \textrm{max}_k \big( X \succcurlyeq C(k,n)\big).
\end{equation*}
So given a $n$-fold equiprobable partition, $k(X,n)$ is the unique maximal positive integer such that $X$ is at least as probable as the union of $k$ cells of our partition. We define
\begin{equation*}
p(X) = \lim_{n \to \infty} \frac{k(X,n)}{n}.
\end{equation*}
One can show that $p$ is a finitely additive probability function which represents $\succcurlyeq$ and that it is unique.
\end{proof}

In this proof, we divide $\Omega$ into more and more fine-grained equiprobable partitions. For every such partition, we `approximate' $p(X)$ by the largest number of cells collectively less probable (according to our comparative probability ordering) than $X$ divided by the number of all cells. Step by step, we get a closer approximation, until we recover the `true' probability of $X$ in the limit. As a simple analogy, think of approximating the area of a two-dimensional figure by drawing more and more fine-grained grids and counting the number of squares covered by the figure divided by the number of all squares. As the grid gets more and more fine-grained, we approximate the area of our figure more and more closely and we recover the true area in the limit.

We can think of the theorem in two ways. If we are inclined towards constructivism, we can think of it as a definition of subjective probability in terms of preferences. In this case, the fact that Ann's subjective probability of rain tomorrow is .3 is \emph{constituted} by the fact that 
\begin{equation*}
\lim_{n \to \infty} \frac{k(rain,n)}{n} = .3,
\end{equation*}
and from this perspective, my axioms are conditions under which subjective probability exists. If we are inclined towards realism, we think that there is a probability function encoding Ann's beliefs not defined in terms of her preferences. From a realist point of view, we can interpret the proof as giving an algorithm to \emph{measure} Ann's subjective probability by constructing better and better approximations. From this perspective, my axioms are conditions under which subjective probability can be measured by this algorithm.

\subsection{Countable Additivity}\label{sec:countable-add}

As it stands, the representation theorem delivers a \emph{finitely additive} probability function which represents our agent's beliefs. This probability function might fail to be \emph{countably additive}.\footnote{The probability function $p : \mathscr{F} \to [0,1]$ is countably additive if for any countable sequence $X_1, X_2, ... $ of pairwise disjoint events in $\mathscr{F}$, $p(\bigcup_{n=1}^{\infty} X_i) = \sum_{n=1}^{\infty} p(X_i)$.} Some decision theorists, for example de Finetti and Savage, have argued that rationality only requires finite additivity and violations of countable additivity are fine. However, there are also reasons to want countable additivity. Most importantly, there are convergence theorems in Bayesian statistics which show that under certain conditions, agents with different priors converge to similar opinions after learning enough shared evidence.\footnote{Subjective Bayesians draw on  such convergence theorems to argue that, despite different priors, rational agents will agree in the long run \parencite[ch. 6]{Earman1992}. Convergence arguments also play an important role in some versions of objective Bayesianism \parencite{Neth2023}.} Many of these convergence theorems require countable additivity \parencite{Elga2016}. So if convergence is a central part of your conception of subjective probability, finitely additive probability is not enough. This is not the place to settle whether arguments for countable additivity are conclusive. The key point is that if countable additivity is desirable, we can add another plausible axiom on preferences to ensure that subjective probabilities are countably additive, building on work by \textcite{Villegas1964}. Details are in the appendix.

\section{Problems Solved}\label{sec:problems-solved}

I explain how my representation theorem does better than Ramsey's method. 

\subsection{Weak Rationality Assumptions}

My axioms do not entail that our agent is an expected utility maximizer. They do not even entail the weaker claim that our agent always respects stochastic dominance.  A quick way to see this is that my axioms only constrain preferences over a very restricted set of acts---two-outcome acts where one outcome is strictly preferred---while expected utility maximization and stochastic dominance constrain preferences over all acts. My axioms do not even require preferences over arbitrary acts to be transitive. So the axioms are compatible with Allais-preferences and violations of stochastic dominance. 

Further, the axioms allow agents to consider many options incomparable. To be sure, \textbf{Non-Degeneracy} requires the existence of at least two comparable outcomes. However, the axioms allow agents to consider \emph{all other} outcomes incomparable. Thus, we can make room for the kind of outcome incomparability discussed above (career as a doctor vs. career as a rock star). Furthermore, we allow agents to consider acts with more than two outcomes incomparable even if all outcomes are comparable, like in the complicated portfolio choice discussed earlier. So, speaking a bit loosely, we allow agents to consider almost all options incomparable.

I still make substantive rationality assumptions. In particular, as discussed above, we can imagine agents which violate \textbf{Certain Prize} and \textbf{Alternative Prize}. It would be desirable to have even more general foundations for subjective probability. But there is a trade-off between substantive rationality axioms which allow us to measure subjective probability but exclude some agents and weak rationality axioms which include these agents but might make measuring subjective probabilities impossible. In particular, Theorem \ref{thm:qualitative-probability} shows that my rationality axioms are necessary conditions for the agent's comparative probability judgments to be representable by a probability function. The comparative probability judgments of agents which violate these axioms cannot be represented by any probability function. So if we want to further weaken these axioms, a fundamentally different approach to measuring subjective probability is needed.

\subsection{Useful Advice}\label{sec:advice}

As explained above, my axioms allow agents to have some irrational preferences. Thus, we can give useful advice. We can define or measure the subjective probabilities of partly irrational agents from their preferences over simple acts and use these probabilities to give useful advice for how to choose among more complicated acts.

You might complain that my axioms are too weak because they only constrain preferences over test acts. But this is a feature, not a bug. We can measure or define subjective probability from preferences over test acts and then apply your favorite decision-theoretic norm to give advice for choices among more complicated acts. I remain agnostic on what exactly this advice looks like. Beyond the basic requirement to respect stochastic dominance, different decision theorists will give different advice: some of them will advise you to maximize expected utility, others will advise you to maximize risk-weighted expected utility and so on. Since I allow incomparable options, there is also the question of how to decide when options are incomparable. But any decision theorist needs to know at least your credences to give useful advice.\footnote{For decision-theoretic advice to be useful, we also need some way to measure your utility function \parencite{Narens2020} and possibly other functions like your risk function \parencite{Neth2019a}.} My representation theorem shows how we can measure or define your credences without already presupposing that your preferences are fully rational and so enables the decision theorist to give useful advice.

Here is a simple toy example for how we can give useful advice. There is an urn with some red marbles, some yellow marbles and some black marbles. A marble will be drawn from this urn.\footnote{To satisfy \textbf{Event Richness}, let us assume that this is the first draw in an infinite sequence of draws from the urn.} We observe that Ann strictly prefers winning one dollar if the marble is red over winning one dollar if the marble is black. So Ann prefers $\{\$1, R; \$0, R^C\}$ over $\{\$1, B; \$0, B^C\}$, where $R$ is the event that the marble is red and $B$ the event that the marble is black. We know that Ann likes more money rather than less, so Ann thinks $R$ is more likely than $B$. Using \textbf{Alternative Prize}, we can infer that Ann must judge $R \cup Y$ to be more likely than $B \cup Y$, where $Y$ is the event that the marble is yellow.

Now suppose Ann faces another choice. The first option pays one dollar if the marble is red, two dollars if the marble is yellow and nothing otherwise: $\{\$1, R; \$2,Y; 0\$, B\}$. The second option pays one dollar if the marble is black, two dollars if the marble is yellow and nothing otherwise: $\{ \$1, B; \$2, Y; \$0,R\}$. We can advise Ann that, to avoid (strict) stochastic dominance, she should prefer the first option. This is genuinely useful advice because it is consistent with my axioms that Ann has no preference among these options or even prefers the stochastically dominated option.

\subsection{No Dependence on Utility}

My axioms make minimal demands on the richness of the outcome space. I only require that there are at least two outcomes our agent is not indifferent between. Thus, we can define or measure subjective probabilities of a very `simple-minded' agent who only distinguishes between the outcome GOOD and the outcome BAD and has an ordinal ranking of these two outcomes. We can disentangle measuring subjective probability from the strong assumptions about value implicit in standard representation theorems.

There is a subtle difference in how my structure axioms compare to Ramsey and Savage. While I do not assume a rich space of outcomes, I do assume a rich space of events, as required by \textbf{Event Richness}. Ramsey's original method does not need such a rich space of events.\footnote{When reconstructing Ramsey's reasoning, \textcite[151]{Fishburn1981} writes that Ramsey assumes `a finite state set'. As noted above, no finite $\Omega$ can satisfy \textbf{Event Richness}.} So in terms of structural richness, my method does better than Ramsey's in one way but worse in another way. This means that our axioms are incomparable in terms of their logical strength. However, from a philosophical point of view, I think that Ramsey and my approach assume a similar amount of structural richness. Ramsey assumes that our agent makes very fine-grained distinctions with respect to the value of outcomes while I assume that our agent makes very fine-grained distinctions with respect to the comparative probability of events. In contrast, my rationality axioms are much weaker than Ramsey's rationality assumptions. Compared with Savage, I make the same structural assumptions about event richness but much weaker rationality assumptions, so we have a strictly more general decision-theoretic foundation for subjective probability than Savage's. 

The upshot: I have shown how to define or measure subjective probability with much weaker rationality axioms than standard representation theorems. If we are interested in ascribing precise subjective probability to partly irrational agents, this is definite progress. One might also take this result as illustration of how strong \textbf{Event Richness} really is. This axiom is doing the heavy lifting in my construction of subjective probability. On the one hand, this might incline some of us to be skeptical of this structure axiom.\footnote{\textcite[p.~98]{Joyce1999} expresses skepticism about the structure axioms in Savage's representation theorem, although one of Joyce's main targets of  completeness which I don't assume.} On the other hand, nobody has figured out how to derive subjective probability without rich preferences and it is probably impossible to do so. So it is fair to say that we have found \emph{better foundations for subjective probability}.

\section{Interpretation}\label{sec:interpretation}

I sketch two ways in which my representation theorem sheds light on the interpretation of subjective probability. First, it naturally fits with a view on which comparative probability is more fundamental than numerical probability. Second, it suggests a subjectivist version of the classical interpretation of probability.

\subsection{Comparativism}

Our starting point were comparative probability judgments defined in terms of preferences. I laid down axioms to ensure that this comparative probability ordering is a qualitative probability and an additional structure axiom to ensure that there is a unique probability function which represents this ordering. While we ultimately end up with a unique probability function which represents our agent's beliefs, this approach naturally suggests a picture on which comparative probability is \emph{more fundamental} than numerical probability.\footnote{Comparativism is discussed by \textcite{Koopman1940, Fine1973, Zynda2000, Hawthorne2016, Stefansson2017, Konek2019, Elliott2022}. Of course, the idea of starting with comparative probability is well-known in decision theory \parencite{Fishburn1986}. However, it is valuable to make explicit that we can be realists about comparative probability and constructivists about numerical probability, while many decision theorists like Savage are constructivist all the way down. As \textcite{Holliday2013} point out, comparative probability can also shed light on probability operators in natural language and so might help us with puzzles about which inferences with these operators are valid \parencite{Yalcin2010, Neth2019b}.} This is in sharp contrast to Ramsey's approach. For Ramsey, subjective probabilities are ratios of utilities and so they are fundamentally quantitative. 

The idea that comparative probability is more fundamental than numerical probability has considerable intuitive appeal. It is more natural to think about which of two events is more likely than to assign numerical probabilities. Furthermore, as I will turn to explain now, taking comparative probability as fundamental allows us to tell a plausible story about the axioms.

My representation theorem naturally fits with a combination of realism and constructivism: realism about comparative probability and constructivism about numerical probability. According to this picture, the comparative probability ordering is psychologically real and not reducible to preferences---rather, preferences serve to measure comparative probability. This is the realist aspect. The advantage of this bit of realism is that we can tell a plausible story about some axioms, in particular \textbf{Non-Degeneracy}, which requires our agent not to be indifferent among all outcomes. It is not very plausible to think that the \emph{existence} of comparative probability requires this axiom, but much more plausible to think that \emph{measuring} comparative probability requires this axiom.

However, in contrast to the comparative probability ordering, the probability function constructed in the representation theorem is not psychologically real but only a `representational device' to talk about the underlying comparative probability ordering. This is the constructivist aspect. The advantage of this bit of constructivism is that we can tell a plausible story about \textbf{Event Richness}. It is implausible to think that subjective probability requires the rich event space postulated by this axiom. If we think of comparative probability as fundamental, we can say that agents might have comparative subjective probabilities even if they do not satisfy this axiom. The axiom describes a condition under which we can \emph{represent} comparative subjective probability by a unique probability function, not a condition under which subjective probability \emph{exists}.

\subsection{Vindicating the Classical Picture}

According to the \emph{classical interpretation of probability} associated with Laplace, we can determine the probability of some event as follows.\footnote{\textcite[ch. 2]{Gillies2000} gives an overview of the classical interpretation and \textcite[ch. 1]{Diaconis2018} briefly recount the history of reasoning about probability in terms of `equally possible' cases.} First, we find a suitable set of `equally possible' cases. Then, we count the number of cases in which the event occurs and divide this number by the number of all cases. For example, if we want to find out the probability of snake eyes (two 1's) when rolling two fair dice, there are 36 `equally possible' cases and exactly one of these cases is snake eyes, so the probability of snake eyes is $\frac{1}{36}$. 

There are many well-known objections to the classical interpretation of probability. First, you might complain that the definition given above is circular. Laplace defines probability in terms of `equally possible' cases, but it is hard to see what `equally possible' could mean other than `equally probable'. Second, the classical interpretation entails that all probabilities are rational, since they are the ratio of two positive integers. But there is nothing incoherent about irrational-valued probabilities.\footnote{\textcite{Hajek1996} uses this as an argument against (finite) frequentism.} Third, what guarantees that we can always find `equally possible' cases? They are easy to find in games of chance but much harder to find in real-life situations, where we might try to find the probability that a nuclear power plant will have a catastrophic accident in the next 100 years \parencite[p.~18]{Halpern2003}.

My representation theorem can be construed as \emph{subjectivist version of the classical interpretation of probability}. Recall how we construct the subjective probability function. To find Ann's subjective probability for rain tomorrow, we find $n$ mutually exclusive and collectively exhaustive events which Ann considers to be equally probable. We write down $k$, the greatest number of events Ann considers to be collectively less likely than rain tomorrow. This is a bit like counting the number of `cases' in which it rains tomorrow. We approximate Ann's subjective probability of rain tomorrow by $k$ divided by the total number of cases. We define Ann's subjective probability of rain tomorrow as the limit of this procedure as $n$ goes to infinity.

The shift towards the subjective helps with some of the well-known worries for the classical interpretation. First, what does it mean to say that the cases are `equally possible'? In our picture, it means that they are judged to be equally probable by our subject and we can know that they are so judged by looking at our subject's preferences. Since we have not defined comparative probability in terms of numerical probability but rather directly in terms of preferences, we can sidestep the circularity worry. Second, since we define subjective probability as limit of a sequence of rational numbers, we can have irrational-valued subjective probabilities. 

Some worries for our subjectivist Laplacean picture still remain. What guarantees that, for any $n$, we can find a partition of $n$ mutually exclusive and collectively exhaustive events which our subject considers to be equally likely? In our construction, this follows from \textbf{Event Richness}, which ensures that the event space of our subject is sufficiently fine-grained. But is it a requirement for the existence of subjective probability to have such a fine-grained event space? Arguably not. If we accept comparativism, we can reply that the more fundamental comparative subjective probabilities still exist without \textbf{Event Richness}, but they may not admit of representation by a unique probability function. From this point of view, Laplace's `equiprobable cases' highlight a condition under which comparative probability judgments can be represented by a unique probability function.

\section{Conclusion}\label{sec:conclusion}

Ramsey wants to reduce subjective probability to preference but makes very demanding rationality assumptions---the agent under consideration has perfectly coherent preferences. I have shown how to provide better foundations for subjective probability: axioms which ensure that there is a unique probability function representing our agent's beliefs while leaving room for mistakes.

Let me close by observing that we if we are convinced that my axioms are rationally required, we can also read my representation theorem as an answer to the question: \emph{Why be probabilistically coherent?} Because rationality requires you to obey the axioms and if you obey the axioms, there is a unique probability function representing your comparative probability judgments. In contrast to other decision-theoretic arguments for probabilistic coherence, such as dutch book arguments or standard representation theorems, this argument does not presuppose or entail expected utility maximization. So we have a new argument for probabilistic coherence from weak assumptions about practical rationality in the face of uncertainty.\footnote{In contrast to accuracy arguments for probabilistic coherence, we also avoid commitments about epistemic value. Furthermore, accuracy arguments run into difficulties when there are infinitely many possibilities \parencite{Kelley2023}.}

\section*{Acknowledgements}

I'm very grateful to Kenny Easwaran, Wesley Holliday, Marcel Jahn, Mikayla Kelley, John MacFarlane, Edward Schwartz and two anonymous referees for helpful comments on previous drafts. Special thanks to Lara Buchak for several rounds of comments and her enthusiasm and support for this project. Further thanks to audiences at the Formal Epistemology Workshop 2022 at UC Irvine and Berkeley's Formal Epistemology Reading Course (FERC) for helpful discussion, especially Mathias B{\"o}hm, Yifeng Ding, and Hanti Lin. While working on this paper, I was supported by a Global Priorities Fellowship by the Forethought Foundation and a Josephine De Karman Fellowship. 

\printbibliography

\section*{Appendix}\label{sec:appendix}

\begin{proposition}
The preference relation $\succsim$ satisfies \textbf{Outcome Independence}, \textbf{Non-Degeneracy}, \textbf{Restricted Ordering}, \textbf{Certain Prize} and  \textbf{Alternative Prize} if and only if the comparative probability ordering $\succcurlyeq$ is a qualitative probability.
\end{proposition}

\begin{proof}
I begin by showing the left-to-right direction. Assume $\succsim$ satisfies the axioms. By \textbf{Outcome Independence} and Definition \ref{def:comparative}, $\succcurlyeq$ is a binary relation on $\mathscr{F}$. By \textbf{Non-Degeneracy}, there are $b,w \in \mathscr{O}$ with $b \succ w$. By \textbf{Restricted Ordering}, for any $X,Y \in \mathscr{F}$, we have $\{b, X; w, X^C\} \succsim \{b, Y; w, Y^C\}$ or $ \{b, Y; w, Y^C\} \succsim \{b, X;w, X^C\}$, so by Definition \ref{def:comparative}, $X \succcurlyeq Y$ or $Y \succcurlyeq X$. Therefore, $\succcurlyeq$ is complete. Analogous reasoning shows that $\succcurlyeq$ is transitive, so $\succcurlyeq$ satisfies Ordering. 

Consider any $X \in \mathscr{F}$. By \textbf{Certain Prize}, we have $\underline{b} \succsim \{b, X; w, X^C\}$ and $\{b, X; w, X^C\} \succsim \underline{w}$. Now $\underline{b} = \{b, \Omega; w, \varnothing\}$ and $\underline{w} = \{w, \Omega;b, \varnothing\}$. So $\{b, \Omega; w, \varnothing\} \succsim \{b, X; w, X^C\}$ and $\{b, X; w, X^C\} \succsim \{w, \Omega;b, \varnothing\}$, and by Definition \ref{def:comparative}, $\Omega \succcurlyeq X \succcurlyeq \varnothing$, so $\succcurlyeq$ satisfies Boundedness. By analogous reasoning, $\succcurlyeq$ satisfies Non-Triviality. 

Now assume $X \cap Z = Y \cap Z = \varnothing$. We want to show that $X \succ Y \iff X \cup Z \succ Y \cup Z$. Assume $X \succ Y$. By Definition \ref{def:comparative},  $\{b, X; w, X^C\} \succsim \{b, Y;w, Y^C\}$ for some $b,w \in \mathscr{O}$ with $b \succ w$. By \textbf{Alternative Prize}, $\{b, X \cup Z; w, (X \cup Z)^{C}\} \succ \{b, Y \cup Z; w, (Y \cup Z)^{C}\}$, so $X \cup Z \succ Y \cup Z$ by Definition \ref{def:comparative}. Analogous reasoning shows the converse implication, so $\succcurlyeq$ satisfies Qualitative Additivity.

I proceed to show the right-to-left direction. Assume $\succcurlyeq$ is a qualitative probability. We want to show that $\succsim$ satisfies the axioms. Assume $\{b, X; w, X^{C}\} \succsim \{b, Y; w, Y^{C}\}$ for some $b,w \in \mathscr{O}$ with $b \succ w$. By Definition \ref{def:comparative}, $X \succcurlyeq Y$. Now assume for \emph{reductio} that for some $b',w' \in \mathscr{O}$ with $b' \succ w'$, $\{b', X; w', X^{C}\} \not \succsim \{b', Y; w', Y^{C}\}$. By Definition \ref{def:comparative}, $X \not \succcurlyeq Y$, which contradicts our assumption. Therefore, $\{b, X; w, X^{C}\} \succsim \{b, Y; w, Y^{C}\}$ for all $b,w \in \mathscr{O}$ such that $b \succ w$, so \textbf{Outcome Independence} holds. We have $\Omega \succ \varnothing$ by Non-Triviality, so $\{b, \Omega; w, \varnothing\} \succ \{b, \varnothing; w, \Omega\}$ for some  $b,w \in \mathscr{O}$ with $b \succ w$. Therefore, there are some $b,w \in \mathscr{O}$ with $b \succ w$, so \textbf{Non-Degeneracy} holds. 

By Ordering, for all $X,Y \in \mathscr{F}$, $X \succcurlyeq Y$ or $Y \succcurlyeq X$. Consider some $b,w \in \mathscr{O}$ with $b \succ w$. We want to show that for all $X,Y \in \mathscr{F}$, either $\{b, X; w, X^{C}\} \succsim \{b, Y; w, Y^{C}\}$ or $\{b, Y; w, Y^{C}\} \succsim \{b, X; w, X^{C}\}$. Assume $X \succcurlyeq Y$. By Definition \ref{def:comparative}, $\{b', X; w', X^{C}\} \succsim \{b', Y; w', Y^{C}\}$ for some  $b',w' \in \mathscr{O}$ with $b' \succ w'$. So by \textbf{Outcome Independence}, $\{b, X; w, X^{C}\} \succsim \{b, Y; w, Y^{C}\}$. An analogous argument applies if $Y \succcurlyeq X$. For transitivity, assume that for some $b,w \in \mathscr{O}$ with $b \succ w$, we have $\{b, X; w, X^{C}\} \succsim \{b, Y; w, Y^{C}\}$ and $\{b, Y; w, Y^{C}\} \succsim \{b, Z; w, Z^{C}\}$. By Definition \ref{def:comparative}, $X \succcurlyeq Y$ and $Y \succcurlyeq Z$, so by Ordering it follows that $X \succcurlyeq Z$. Again by Definition \ref{def:comparative},  $\{b', X; w', X^{C}\} \succsim \{b', Y; w', Y^{C}\}$ for some  $b',w' \in \mathscr{O}$ with $b' \succ w'$. By \textbf{Outcome Independence}, $\{b, X; w, X^{C}\} \succsim \{b, Y; w, Y^{C}\}$, so \textbf{Restricted Ordering} holds.  By Boundedness, for all $X \in \mathscr{F}$, $\Omega \succcurlyeq X$ and $X \succcurlyeq \varnothing$, so $\underline{b} \succsim \{b, X; w, X^C\}$ and $\{b, X; w, X^C\} \succsim \underline{w}$ for all $b,w \in \mathscr{O}$ with $b \succ w$ so \textbf{Certain Prize} holds. 

Finally, let $X \cap Z = Y \cap Z = \varnothing$ and assume $\{b, X; w, X^{C}\} \succ \{b, Y; w, Y^{C}\}$ for some $b,w \in \mathscr{O}$ with $b \succ w$. By Definition \ref{def:comparative}, $X \succ Y$, so by Qualitative Additivity, $X \cup Z \succ Y \cup Z$. Again by Definition \ref{def:comparative} and \textbf{Outcome Independence}, $\{b, X \cup Z; w, (X \cup Z)^{C}\} \succ \{b, Y \cup Z; w, (Y \cup Z)^{C}\}$. A similar argument shows the converse entailment, so \textbf{Alternative Prize} holds.
\end{proof}

\noindent We can ensure countable additivity by adding this axiom:
\begin{quote}
\textbf{Monotone Preference Continuity}. For any $b, w \in \mathscr{O}$ with $b \succ w$ and any monotonically increasing sequence of events $X_1 \subseteq X_2  \subseteq ... $ with $\bigcup_{n=1}^{\infty} X_i = X$, if for all $n$, $\{b, Y; w, Y^C\} \succsim \{b, X_n; w, X_n ^C\}$, then $\{b, Y; w, Y^C\} \succsim \{b, X; w, X^C\}$.
\end{quote}
Building on work by \textcite{Villegas1964}, we can show:
\begin{theorem}
If the preference relation $\succsim$ satisfies \textbf{Outcome Independence}, \textbf{Non-Degeneracy}, \textbf{Restricted Ordering}, \textbf{Certain Prize}, \textbf{Alternative Prize}, \textbf{Event Richness} and \textbf{Monotone Preference Continuity},  there is a unique countably additive probability function representing $\succcurlyeq$.
\end{theorem}
\begin{proof}
Assume the preference relation $\succsim$ satisfies \textbf{Outcome Independence}, \textbf{Non-Degeneracy}, \textbf{Restricted Ordering}, \textbf{Certain Prize}, \textbf{Alternative Prize} and \textbf{Event Richness}. Then $\succcurlyeq$ is  \emph{atomless}, which means that for every $X \succ \varnothing$, there is some $Y \subseteq X$ such that $X \succ Y \succ \varnothing$. Now, given \textbf{Monotone Preference Continuity}, $\succcurlyeq$ satisfies:
\begin{quote}
\textbf{Monotone Probability Continuity.} If $X_1, X_2, ...$ is a monotonically increasing sequence of events with $\bigcup_{n=1}^{\infty} X_i = X$, and for every $n$, $Y \succcurlyeq X_n$, then $Y \succcurlyeq X$.
\end{quote}
\textcite{Villegas1964} shows that if $\succcurlyeq$ is an atomless qualitative probability and \textbf{Monotone Probability Continuity} holds, there is a unique countably additive probability function representing $\succcurlyeq$. By Theorem \ref{thm:representation}, there is a unique probability function representing $\succcurlyeq$. By Villegas' result, \textbf{Monotone Preference Continuity} implies that this probability function must be countably additive.
\end{proof}

\end{document}